\UseRawInputEncoding % arxiv
\documentclass[1p,final]{elsarticle}
\usepackage{amsfonts,color,morefloats,pslatex}
\usepackage{amssymb,amsthm, amsmath,latexsym}

\allowdisplaybreaks[4]

\newtheorem{rem}{Remark}
\newtheorem{theorem}{Theorem}
\newtheorem{lemma}[theorem]{Lemma}

\newtheorem{open}{Open Problem}

\newtheorem{example}{Example}
\newtheorem{conj}[theorem]{Conjecture}

\newcommand{\tr}{{\mathrm{Tr}}}

\newcommand{\gf}{{\mathrm{GF}}}

\newcommand{\cP}{{\mathcal{P}}}
\newcommand{\cB}{{\mathcal{B}}}

\newcommand{\C}{{\mathcal{C}}}

\newcommand{\bD}{{\mathbb{D}}}

\begin{document}

\begin{frontmatter}



\title{ Codes and Pseudo-Geometric Designs from the Ternary $m$-Sequences with Welch-type decimation $d=2\cdot 3^{(n-1)/2}+1$
\tnotetext[fn1]{C. Xiang's research was supported by
the National Natural Science Foundation of China under grant numbers 12171162 and 11971175; C. Tang's research was supported by the National Natural Science Foundation of China under grant number 12231015, the Sichuan Provincial Youth Science and Technology Fund under grant number 2022JDJQ0041 and the Innovation Team Funds of China West Normal University under grant number KCXTD2022-5; H. Yan's research was supported by the Natural Science Foundation of Sichuan Province under grant number 2022NSFSC1805 and the
Fundamental Research Funds for the Central Universities of China under grant number 2682023ZTPY002.}}

%\author{Cunsheng Ding}
%\ead{cding@ust.hk}

\author[cx]{Can Xiang}
\address[cx]{College of Mathematics and Informatics, South China Agricultural University, Guangzhou, Guangdong 510642, China}
\ead{cxiangcxiang@hotmail.com}
\author[cmt]{Chunming Tang}
\address[cmt]{School of Information Science and Technology, Southwest Jiaotong University, Chengdu 610031, China}
\ead{tangchunmingmath@163.com}
\author[hdy]{Haode Yan}
\address[hdy]{School of Mathematics, Southwest Jiaotong University, Chengdu 610031, China}
\ead{hdyan@swjtu.edu.cn}
\author[mg]{Min Guo}
\address[mg]{School of Mathematics and Information,China West Normal University,Nanchong, Sichuan 637002, China}
\ead{1289042385@qq.com}

%\cortext[zz]{ is the Corresponding author}

%\address[cding]{Department of Computer Science and Engineering, The Hong Kong University of Science and Technology, Clear Water Bay, Kowloon, Hong Kong, China}

\begin{abstract}
Pseudo-geometric designs are combinatorial designs which share the same parameters as
a finite geometry design, but which are not isomorphic to that design.
As far as we know, many pseudo-geometric designs have been constructed by the methods of finite geometries and combinatorics. However, none of pseudo-geometric designs with the parameters $S\left (2, q+1,(q^n-1)/(q-1)\right )$ is constructed by the approach of coding theory. In this paper, we use cyclic codes to construct pseudo-geometric designs. We firstly present a family of ternary cyclic codes from the $m$-sequences with Welch-type decimation $d=2\cdot 3^{(n-1)/2}+1$, and obtain some infinite family of 2-designs and a family
of Steiner systems $S\left (2, 4,  (3^n-1)/2\right )$ using these cyclic codes and their duals. Moreover, the parameters of these cyclic codes and their shortened codes are also determined. Some of those ternary codes are optimal or almost optimal. Finally, we show that one of these obtained Steiner systems is inequivalent to the point-line design of the projective space $\mathrm{PG}(n-1,3)$ and thus is a pseudo-geometric design.

%\textcolor{blue}{}

%There are several approaches to constructing linear codes over finite fields. One is a sequence approach, which has
%been intensively investigated in the past decade. The objective of this paper is to construct linear codes associated with the m-Sequences with Welch-type decimation.
%A class of ternary linear codes with three weights are obtained and their parameters are determined. Some of those ternary codes are optimal or almost optimal. Furthermore, we show that those codes hold $2$-designs.

\end{abstract}


\end{frontmatter}

\section{Introduction}\label{sec-int}

Let $p$ be a prime and $q = p^m$ for some positive integer $m$. Let $\gf(q)$ be the finite field
of cardinality $q$. A $[v,\, k,\,d]$ linear code $\C$ over $\gf(q)$ is a $k$-dimensional subspace of $\gf(q)^v$ with minimum (Hamming) distance $d$.
A $[v,\, k,\,d]$ linear code $\C$ is said to be {\em cyclic} if
$(c_0,c_1, \cdots, c_{v-1}) \in \C$ implies $(c_{v-1}, c_0, c_1, \cdots, c_{v-2}) \in \C$. A $[v,\, k,\,d]$ linear code $\C$ is said to be
\emph{distance-optimal}, \emph{dimension-optimal} and \emph{length-optimal} if there is no $[v,k,d+1]$ code exists, no $[v,k+1,d]$ code exists and no $[v',k,d]$ code exists with $v' < v$, respectively. A code is said to be \emph{optimal} if it is distance-optimal, or dimension-optimal, or length-optimal, or meets a
bound for linear codes. A$[v,k,d]$ code is said to be \emph{almost optimal} if an $[v,k+1,d]$, or $[v,k,d+1]$, or $[v-1,k,d]$ code is optimal.

%A  $[v,k,d]$ code over $\gf(q)$ is said to be \emph{distance-optimal} if no $[v,k,d']$ code over $\gf(q)$ with $d'>d$ exists, \emph{dimension-optimal} if no $[v,k',d]$ code over $\gf(q)$ with $k'>k$ exists, and \emph{length-optimal} if no $[v',k,d]$ code over $\gf(q)$ with $v'< v$ exists. A linear code is said to be optimal if it is distance-optimal, or dimension-optimal, or length-optimal, or meets a bound for linear codes.

It is known that there are several approaches to constructing linear codes over finite fields (see, for example, \cite{ding2018c,FLZ2016,LDT2020,lisihem2020,sihem2017}).
One is a sequence approach and it has been intensively investigated for constructing cyclic codes which are a subclass of linear codes. All cyclic codes over finite fields can be produced by using the sequence approach. Ding employed the two-prime sequences and the cyclotomic sequences of order four to construct
cyclic codes from this sequence approach in \cite{ding20120} and \cite{ding20130} respectively. The lengths and dimensions of
the cyclic codes obtained in \cite{ding20120,ding20130} were determined and some lower bounds on their minimal distances
were also developed. Ding also showed that the obtained cyclic codes are very good in general, and some of them are (almost) optimal. In addition, sequences with nice property can be obtained from some cryptographic functions over finite
fields, and it seems that researchers prefer to use sequences associated with cryptographic functions
to construct cyclic codes (see, for example, \cite{ding2012,dingzhou2014,TYX2014}), since the cryptographic properties of the employed functions can be used to
determine the parameters of the cyclic codes. The reader is referred to \cite{ding2018c,lisihem2020} for further information on the
sequence construction of cyclic codes.

The sequence generated by the non-degenerate linear feedback shift register with the largest period is called the \emph{$m$-sequence}. As special sequences, $m$-sequences have wide
applications in constructing linear codes and the parameters of these codes can be obtained from the properties of the cross correlation function between two $m$-sequences.
%It is notice that some linear codes with good parameters and their parameters can be obtained from the property of the cross correlation function between two m-Sequences.
Thus, some works focus on the determination of
the property of the cross correlation function between two $m$-sequences generated by different characteristic polynomials (see, for example, \cite{CCD,CD,Gold,H1976,HX,Ka,Tra,xia2014}).
However, it is very difficult to determine the value distribution of the cross correlation function in general. In fact, only a small number of $p$-ary $m$-sequences with few-valued
cross-correlation function for $p$ odd prime are reported in the literature. Note that for $p$ odd prime, the three-valued distribution of the cross correlation function from some $p$-ary $m$-sequences $\{s(t)\}$ with period $p^n-1$ and different decimation $d$ were determined and listed
as follows:
%below.
\begin{enumerate}
  \item [(\uppercase\expandafter{\romannumeral1})]  $\frac{n}{s}\geq 3$ is odd, $d=(p^{3t}+1)/(p^t+1)$ \cite{H1976};
  \item [(\uppercase\expandafter{\romannumeral2})]  $\frac{n}{s}\geq 3$ is odd, $d=(p^{2t}+1)/2$ \cite{H1976};
  \item [(\uppercase\expandafter{\romannumeral3})] $p=3$, $n$ is odd, $d=2\cdot 3^\frac{n-1}{2}+1$ \cite{Dob2001};
 % \item [(\uppercase\expandafter{\romannumeral4})] $\frac{n}{s}\geq 3$ be odd, $d(p^t+1)\equiv 2~(mod~p^n-1)$, $d \equiv 1~(mod~p^n-1)$ \cite{xia2014};
\end{enumerate}
where $s=\gcd(t,n)$ and $\gcd(d,p^n-1)=1$.
Specially, Dobbertin et al. (i.e., see the above case (\uppercase\expandafter{\romannumeral3}) ) showed that the cross correlation function between two
ternary $m$-Sequences with period $3^n-1$ and Welch-type decimation $d=2\cdot 3^m+1$ was three-Valued and preferred, where $n=2m+1$ and $q=3^n$. By the property of this cross correlation function, the parameters of the corresponding code
\begin{eqnarray}\label{c0}
\{(\tr(ax+bx^d))_{x\in \gf(3^n)}|a,b\in \gf(3^n)\}
\end{eqnarray}
were also determined, where  $\tr$ is the trace function from $\gf(3^n)$ to $\gf(3)$. This code has parameters $[q, 2n, 2(3^{n-1}-3^m)]$ and three nonzero weights, i.e., $2(3^{n-1}-3^m)$,  $ 2\cdot 3^{n-1}$ and $2(3^{n-1}+3^m)$.
In this paper, we will consider a class of cyclic codes
\begin{eqnarray}\label{c1}
\widetilde{\C_{n}}=\{(\tr(a \alpha^{2i}+b\alpha^{2di}))_{i=0}^{(3^n-3)/2}|a,b\in \gf(3^n)\}
\end{eqnarray}
from the above ternary codes associated with the ternary $m$-Sequences,
where $\alpha$ is a primitive element of $\mathrm{GF(3^n)}$ and $\tr$ is the trace function from $\gf(3^n)$ to $\gf(3)$. The first objective of this paper is to show that the code $\widetilde{\C_{n}}$ and its dual code hold $2$-designs and determine their parameters. Moreover, the supports of the minimum weight codewords in the dual of $\widetilde{\C_{n}}$ form a Steiner system.

%We finally show that one of these Sterner system is a pseudo-geometric design with parameter $S(2,4,121)$.

It is well known that the number of non-isomorphic designs having the same parameters as $\mathrm{PG}_1(n,q)$ grows exponentially with linear growth of $n$ \cite{Jungnickel-recent,Wilson1975}. In 2010, Jungnickel \cite{Jungnickel2010} showed that every permutation of a hyperplane of $\mathrm{PG}_1(n-1,q)$ gives rise to a design with the same parameters as $\mathrm{PG}_1(n-1,q)$. However, none of pseudo-geometric designs with the same parameters as $\mathrm{PG}_1(n,q)$ is constructed by using the method of coding theory. The second objective of this paper is to show that the Steiner system $S(2,4,121)$ constructed by using this method is a pseudo-geometric design. It will provide theories support for constructing pseudo-geometric designs by using coding theory.

The rest of this paper is arranged as follows. Section \ref{sec-pre} states some notation and results
about linear codes and combinatorial $t$-designs. The parameters of the ternary code $\widetilde{\C_{n}}$ and its dual are determined in Section \ref{sec-3code}.
Section \ref{sec-des2} gives some infinite families of $2$-designs and determines their parameters. We finally show that the Steiner system $S(2,4,121)$ presented in this paper
is inequivalent to the classic point-line design $\mathrm{PG}_1(4,3)$ of the projective space $\mathrm{PG}(4,3)$  in Section \ref{sec:pgc}. The conclusion of this paper is given in Section  \ref{sec-summary}.

\section{Preliminaries}\label{sec-pre}

%In this section, we state some notation and recall some basic facts which will be used in the rest of this paper.

In this section, some notation and basic facts are described and will be needed later.

\subsection{Some results of linear codes}

Let $\C$ be a $[v,k,d]$ linear code over $\gf(q)$. Let $A_i$ denote the number of codewords with Hamming weight $i$ in a code
$\C$ of length $v$. The weight enumerator of $\C$ is defined by
$
1+A_1z+A_2z^2+ \cdots + A_v z^v.
$
The sequence $(1,A_1,\ldots,A_v)$ is called the weight distribution of $\C$ and is an important research topic in coding theory,
as it contains crucial information about the error correcting capability of the code. Thus the study of the weight distribution
has attracted much attention in coding theory and much work focuses on the determination of
the weight distributions of linear codes (see, for example, \cite{ding2018,Ding16,sihem2020,sihem2017,zhou20131}). A code $\C$ is said to be a $t'$-weight code  if the number of nonzero $A_i$ in the sequence $(A_1, A_2, \cdots, A_v)$ is equal to $t'$.
Denote by $\C^\bot$  and $(A_0^{\perp}, A_1^{\perp}, \dots, A_\nu^{\perp})$ the dual code of a linear code $\C$ and its weight distribution, respectively.
The \emph{Pless power moments} \cite{HP10}
play an important role in calculating
the weight distributions of linear codes. The first five Pless power moments identities are given by
\begin{align}\label{eq:PPM}
 & \sum_{i=0}^\nu  A_i= q^k, \nonumber \\
 & \sum_{i=0}^\nu  i\cdot A_i= q^{k-1} (qv-v-A_1 ^\perp), \nonumber \\
 & \sum_{i=0}^\nu  i^2 \cdot A_i= q^{k-2} \left [(q-1)v(qv-v+1)-(2qv-q-2v+2)A_1^\perp +2 A_2^\perp \right], \nonumber \\
 & \sum_{i=0}^\nu  i^3 \cdot A_i= q^{k-3}[ (q-1)v (q^2 v^2 -2q v^2 + 3q v - q + v^2 - 3 v + 2) - (3q^2 v^2 -3q^2 v - 6q v^2  \nonumber \\
 &  ~~~~~~~~~~~~~~  + 12q v + q^2 - 6q + 3v^2 - 9v + 6) A_1^\perp  + 6(qv- q- v +2)A_2^\perp -6 A_3 ^\perp  ], \nonumber \\
 & \sum_{i=0}^\nu  i^4 \cdot A_i= q^{k-4}[  (q-1)v (q^3 v^3 -3q^2 v^3 +6q^2 v^2 -4q^2 v + q^2+3 q v^3 - 12q v^2 +15qv -6q-v^3  \nonumber \\
 &  ~~~~~~~~~~~~~~                        +6v^2-11v +6)  - ( 4q^3 v^3-6q^3v^2 + 4q^3 v -q^3 - 12q^2 v^3+ 36q^2 v^2 - 38q^2 v + 14q^2  \nonumber \\
 &   ~~~~~~~~~~~~~~  +12q v^3 - 54q v^2 + 78q v -36q - 4v^3 + 24v^2 - 44v+ 24               ) A_1^\perp       \nonumber \\
 &    ~~~~~~~~~~~~~~ +(12q^2v^2 -24q^2v + 14q^2 -24qv^2+ 84qv - 72q + 12v^2 -60v + 72) A_2^\perp \nonumber \\
 &  ~~~~~~~~~~~~~~ - (24qv - 36q -24v + 72)A_3^\perp +24 A_4 ^\perp    ].
 \end{align}

%Let $s=(s_i)_{i=0}^{\infty}$ be a sequence of period $n$ over $\gf(q)$.

%It is well known that some linear codes with good parameters can be obtained from periodic sequences. The sequence generated by the non-degenerate linear feedback shift register with the largest period is called the $m$-sequence. Let $s=(s_i)_{i=0}^{\infty}$ be a $m$-sequence of period $p^n-1$ over $\gf(q)$.

\subsection{Combinatorial t-designs and some related results}

Let $k$, $t$ and $v$ be positive integers with $1 \leq t \leq k \leq  v$. Let $\cP$ be a set with $v$ elements and $\cB$ be a set of some $k$-subsets of
$\cP$. $\cB$  is called the point set and  $\cP$ is called the block set in general. The incidence structure
$\bD = (\cP, \cB)$ is called a $t$-$(v, k, \lambda)$ {\em design\index{design}} (or {\em $t$-design\index{$t$-design}}) if each $t$-subset of $\cP$ is contained in exactly $\lambda$ blocks of
$\cB$.
Let $\binom{\cP}{k}$ denote the set consisting of all $k$-subsets of the point set $\cP$. Then the incidence structure $(\cP, \binom{\cP}{k})$ is a $k$-$(v, k, 1)$ design and is called a \emph{complete design}. The special incidence structure $(\cP, \emptyset)$ is called a $t$-$(v, k, 0)$ trivial design
for all $t$  and $k$ . A combinatorial $t$-design is said to be {\em simple\index{simple}} if its block set $\cB$ does not have
a repeated block. When $t \geq 2$ and $\lambda=1$, a $t$-$(v,k,\lambda)$ design is called  a
{\em Steiner system\index{Steiner system}}
and denoted by $S(t,k, v)$. The parameters of a combinatorial $t$-$(v, k, \lambda)$ design must satisfy the following equation:
\begin{eqnarray}\label{eq:bb}
b  =\lambda \frac{\binom{v}{t}}{\binom{k}{t}}
\end{eqnarray}
where $b$ is the cardinality of $\cB$.

Let $(\mathcal P, \mathcal B)$ and $(\mathcal P', \mathcal B')$ be two $t$-designs.
An isomorphism from $(\mathcal P, \mathcal B)$ to
$(\mathcal P', \mathcal B')$ is a $1$-$1$ mapping $\sigma: \mathcal P \rightarrow \mathcal P'$ such that $\widetilde{\sigma}(\mathcal P)= \mathcal P'$ (where $\widetilde{\sigma}$ is the mapping induced on blocks by $\sigma$).
Isomorphism is an equivalence relation, and the $t$-designs $(\mathcal P, \mathcal B)$ and $(\mathcal P', \mathcal B')$ are isomorphic. An automorphism of $(\mathcal P, \mathcal B)$ is an isomorphism of $(\mathcal P, \mathcal B)$ to itself.
The set of all such automorphisms forms a group under composition
called the automorphism group of the $t$-design. A $t$-$(v, k, \lambda)$ design is cyclic if it has an automorphism consisting of a single cycle of length $v$
.

%The interplay between linear codes and $t$-designs has attracted a lot of attention for both directions. It is well known that the supports of all codewords with a fixed weight in a code may hold a t-design.
%The interplay between linear codes and combinatorial $t$-designs has attracted a lot of attention.

It is well known that $t$-designs and linear codes  are interactive with each other.
A $t$-design $\mathbb  D=(\mathcal P, \mathcal B)$ can be used to construct a linear code over GF($q$) for any $q$ (see, for example, \cite{Dingt20201,ton1,ton2}).
Meanwhile, a linear code $\C$ may produce a $t$-design which is formed by supports of codewords of a fixed Hamming weight in $\C$. Let $\nu$ be the length of $\mathcal C$ and the set of the coordinates of codewords in $\mathcal C$ is denoted by $\mathcal P(\mathcal C)=\{0,1, 2, \dots, \nu-1\}$. The \emph{support} of  $\mathbf c$
is defined by
\begin{align*}
\mathrm{Supp}(\mathbf c) = \{i: c_i \neq 0, i \in \mathcal P(\mathcal C)\}
\end{align*}
for any codeword $\mathbf c =(c_0, c_1, \dots, c_{\nu-1})$ in $\mathcal C$.
Let $\mathcal B_{w}(\mathcal C)$ denote the set $\{\{   \mathrm{Supp}(\mathbf c): wt(\mathbf{c})=w
~\text{and}~\mathbf{c}\in \mathcal{C}\}\}$, where $\{\{\}\}$ is the multiset notation. For some special code $\mathcal C$,
the incidence structure $\left (\mathcal P(\mathcal C),  \mathcal B_{w}(\mathcal C) \right)$
could be a $t$-$(v,w,\lambda)$ design for some positive integers $t$ and $\lambda$.
We say that the code $\mathcal C$ \emph{supports $t$-designs} if $\left (\mathcal P(\mathcal C),  \mathcal B_{w}(\mathcal C) \right)$ is a $t$-design for all $w$ with $0\le w \le \nu$. By definition, such design
$\left (\mathcal P(\mathcal C),  \mathcal B_{w}(\mathcal C) \right)$ could have some repeated
blocks, or could be simple, or may be trivial.
In this way, many $t$-designs have been constructed from linear codes (see, for example, \cite{Ding18dcc,ding2018,Tangding2020,du1,Tangdcc2019}). A major way to construct combinatorial $t$-designs with linear codes over finite fields is the use of linear codes with $t$-transitive or $t$-homogeneous automorphism groups (see \cite[Theorem 4.18]{ding2018}) and some combinatorial $t$-designs (see, for example, \cite{LiuDing2017,Liudingtang2021}) were obtained by this way.
Another major way to construct $t$-designs with linear codes is the use of the
Assmus-Mattson Theorem (AM Theorem for short) in \cite[Theorem 4.14]{ding2018} and the generalized version of the
AM Theorem in \cite{Tangit2019}, which was recently employed to construct a number of $t$-designs (see, for example, \cite{ding2018,du1}).
The following theorem is a generalized version of the
AM Theorem, which was developed in \cite[Theorem 5.3]{Tangit2019} and will be needed in this paper.

\begin{theorem}\cite{Tangit2019}\label{thm-designGAMtheorem}
Let $\mathcal C$ be a linear code over $\mathrm{GF}(q)$ with minimum distance $d$ and length $\nu$.
Let $\mathcal C^{\perp}$ denote the dual of $\mathcal C$ with minimum distance $d^{\perp}$.
Let $s$ and $t$ be positive integers with $t< \min \{d, d^{\perp}\}$. Let $S$ be a $s$-subset
of the set $\{d, d+1, d+2, \ldots, \nu-t  \}$.
Suppose that
$\left ( \mathcal P(\mathcal C), \mathcal B_{\ell}(\mathcal C) \right )$ and $\left ( \mathcal P(\mathcal C^{\perp}), \mathcal B_{\ell^{\perp}}(\mathcal C^{\perp}) \right )$
are $t$-designs  for
$\ell    \in \{d, d+1, d+2, \ldots, \nu-t  \} \setminus S $ and $0\le \ell^{\perp} \le s+t-1$, respectively. Then
the incidence structures
 $\left ( \mathcal P(\mathcal C) , \mathcal B_k(\mathcal C) \right )$ and
  $\left ( \mathcal P(\mathcal C^{\perp}), \mathcal B_{k}(\mathcal C^{\perp}) \right )$ are
  $t$-designs for any
$t\le k  \le \nu$, and particularly,
\begin{itemize}
\item the incidence structure $\left ( \mathcal P(\mathcal C) , \mathcal B_k(\mathcal C) \right )$ is a simple $t$-design
      for all integers $k$ with $d \leq k \leq w$, where $w$ is defined to be the largest  integer
      such that $w \leq \nu$ and
      $$
      w-\left\lfloor \frac{w+q-2}{q-1} \right\rfloor <d;
      $$
\item  and the incidence structure $\left ( \mathcal P(\mathcal C^{\perp}), \mathcal B_{k}(\mathcal C^{\perp}) \right )$ is
       a simple $t$-design
      for all integers $k$ with $d \leq k \leq w^\perp$, where $w^\perp$ is defined to be the largest integer
      such that $w^\perp \leq \nu$ and
      $$
      w^\perp-\left\lfloor \frac{w^\perp+q-2}{q-1} \right\rfloor <d^\perp.
      $$
\end{itemize}
\end{theorem}

Let $\mathcal C$ be a linear code over $\mathrm{GF}(q)$ and $T$ be a set of $t$ coordinate positions in $\C$. We puncture $\mathcal  C$  on
$T$ and obtain a linear code which is called the \emph{punctured code}  of $\mathcal C$ on $T$ and denoted by $\mathcal  C^T$.
We use $\mathcal C(T)$ to denote the set of codewords that are
$\mathbf{0}$ on $T$.
We now puncture $\mathcal C(T)$ on $T$, and obtain a linear code $\mathcal C_{T}$, which is called the \emph{shortened code} of $\mathcal C$ on $T$.
The following property plays an important role in determining the parameters of some shortened codes of $\mathcal{C}$ supporting $t$-designs  in \cite[Theorem 3.2]{Tangit2019}.

\begin{theorem}\cite{Tangit2019}\label{thm-PS}
Let $\mathcal C$ be a $[\nu, m, d]$ linear code  over $\mathrm{GF}(q)$ and  $d^{\perp}$  the minimum distance of $\mathcal  C^{\perp}$.
 Let  $t$ be a positive integer with  $0< t <\min \{d, d^{\perp}\}$.
Let $T$ be  a  set of $t$ coordinate positions in $\mathcal  C$.
Suppose that $\left ( \mathcal P(\mathcal C) , \mathcal B_i(\mathcal C) \right )$ is a $t$-design for any $i$ with $d \le i \le \nu-t$.
Then the shortened code $\mathcal C_T$ is a linear code of length $\nu-t$ and dimension $m-t$. The weight distribution
$\left ( A_k(\mathcal C_T) \right )_{k=0}^{\nu-t}$ of $\mathcal C_T$ is independent of the specific choice of the elements
in $T$. Specifically,
$$A_k(\mathcal C_T) =\frac{ \binom{k}{t} \binom{\nu-t}{k}}{ \binom{\nu }{t} \binom{\nu-t}{k-t}}A_k(\mathcal C).$$
\end{theorem}

\section{A class of ternary linear codes} \label{sec-3code}

In this section, our task is to determine the parameters of the ternary cyclic code $\widetilde{\C_{n}}$ and its dual code.
%In addition,  we will consider some shortened codes of $\widetilde{\C_{n}}$ and determine their parameters. Some of these codes are optimal or almost optimal.
%To be end, we need the results in the following two lemmas. The former was documented in \cite{Dob2001} and the latter is given in Lemma \ref{A230}.

Starting from this section till the end of this paper unless otherwise stated, we assume that $n=2m+1$ with some positive integer $m$, $q=3^n$, $d=2\cdot 3^m+1$ and $SQ$ (resp., $ NSQ$) is the set of all squares (resp., non-squares) in $\gf(q)$.

In order to determine the parameters of the ternary code $\widetilde{\C_{n}}$ and its dual code, we need the results in the following two lemmas. The former was documented in \cite{Dob2001} and the latter is given in Lemma \ref{A230}.

\begin{lemma}\cite{Dob2001}\label{helle8}
Let symbols and notation be the same as before. Let $m\geq 2$ be an positive integer. Then the system of equations
\begin{eqnarray*}
\left\{
\begin{array}{ll}
1+y_1+y_2=0&\\[2mm]
1+ y_1^{~d}+y_2^{~d}=0
\end{array}
 \right..
\end{eqnarray*}
for $y_1,y_2\in \gf(q)^*$ has only one solution, i.e., $y_1=y_2=1$.

\end{lemma}

\begin{lemma}\label{A230}
Let symbols and notation be the same as before.
Let $m\geq 2$ be an positive integer and $c_j \in \gf(3)^*$ for any positive integer $j$.  Let $N$ be the number of $\{x_1,x_2,\cdots,x_i\}\subseteq SQ$ satisfying the system of equations
\begin{eqnarray}\label{numA23}
\left\{
\begin{array}{ll}
\sum_{s=1}^{s=i} c_s x_s=0&\\[2mm]
\sum_{s=1}^{s=i} c_s x_s^{~d}=0
\end{array}
 \right.
\end{eqnarray}
where $i\in\{1, 2,3\}$. Then $N=0$.
\end{lemma}

\begin{proof}
It is clear that $N=0$ for $i=1$. Next we give the proof for the case $i=2$ and $i=3$, respectively.

If $i=2$, (\ref{numA23}) becomes
\begin{eqnarray}\label{numA2}
\left\{
\begin{array}{ll}
c_1x_1+c_2x_2=0&\\[2mm]
c_1 x_1^{~d}+c_2 x_2^{~d}=0
\end{array}
 \right..
\end{eqnarray}
Note that $c_1,c_2 \in \gf(3)^*$, i.e., $c_1,c_2 \in \{1,-1\}$. It is obvious that $x_1=x_2$ if $c_1 \neq c_2$.
When $c_1=c_2$, (\ref{numA2}) is equivalent to
\begin{eqnarray}\label{numA2-1}
\left\{
\begin{array}{ll}
x_1+x_2=0&\\[2mm]
x_1^{~d}+x_2^{~d}=0
\end{array}
 \right..
\end{eqnarray}
Since $d$ is odd, from (\ref{numA2-1}) we have $x_1=-x_2$. This means that one of $x_1$ and $x_2$ is in $SQ$ and the other is in $NSQ$. Thus, $N=0$.

If  $i=3$, (\ref{numA23}) becomes
\begin{eqnarray}\label{numA3}
\left\{
\begin{array}{ll}
c_1x_1+c_2x_2+c_3 x_3=0&\\[2mm]
c_1 x_1^{~d}+c_2 x_2^{~d}+c_3 x_3^{~d}=0
\end{array}
 \right..
\end{eqnarray}
Note that $c_1,c_2, c_3 \in \gf(3)^*$, i.e., $c_1,c_2, c_3 \in \{1,-1\}$.
Due to symmetry, it is sufficient to consider the following two
cases.
\begin{itemize}
  \item When $c_1=c_2=c_3=1$, (\ref{numA3}) becomes
\begin{eqnarray}\label{numA3-1}
\left\{
\begin{array}{ll}
1+x_2/x_1+x_3/x_1=0&\\[2mm]
1+ (x_2/x_1)^{~d}+(x_3/x_1)^{~d}=0
\end{array}
 \right..
\end{eqnarray}
Putting $y_1=x_2/x_1$ and $y_2=x_3/x_1$ into (\ref{numA3-1}) leads
to
\begin{eqnarray}\label{numA3-11}
\left\{
\begin{array}{ll}
1+y_1+y_2=0&\\[2mm]
1+ y_1^{~d}+y_2^{~d}=0
\end{array}
 \right..
\end{eqnarray}
From Lemma \ref{helle8} and (\ref{numA3-11}), we then deduce that $x_2/x_1=x_3/x_2=1$, i.e., $x_1=x_2=x_3$. This means that there is no set $\{x_1,x_2,x_3\}\subseteq SQ$ satisfying (\ref{numA3-1}). Thus, $N=0$.

  \item When $c_1=c_2=1, c_3=-1$. Substituting $y_1=x_2/x_1$ and $y_2=x_3/x_1$ into (\ref{numA3}) leads
to
$$
(y_1^{~3^m}-1) (y_1^{~3^m}-y_1)=0,
$$
which means that $y_1\in \gf(3^n)\bigcap \gf(3^m)=\gf(3)$ since $\gcd(n,m)=1$. Note that $y_1,y_2\neq 0$. Then $(y_1, y_2)=(1,-1)$ or  $(-1,1)$.
Thus, there is no set $\{x_1,x_2,x_3\}\subseteq SQ$ satisfying $\{x_2/x_1, x_3/x_1\} =\{1,-1\}$.
\end{itemize}

The desired conclusions then follow from the discussions above. This completes the proof.
\end{proof}

Next we give the parameters of $\widetilde{\C_{n}}$ and its dual $\widetilde{\C_{n}}^\bot$ in the following theorem, which is one of the main results in this paper.

\begin{theorem}\label{thm-main1}
Let symbols and notation be the same as before. Let $m\geq 2$ be an positive integer and $\C=\widetilde{\C_{n}}$ be defined by (\ref{c1}).
Then we have the following results.
\begin{enumerate}
 \item [(\uppercase\expandafter{\romannumeral1})] The code $\C$  has parameters $[(q-1)/2, 2n, 3^{n-1}-3^m]$  and  the weight distribution in Table \ref{tab-3}.
 \item [(\uppercase\expandafter{\romannumeral2})] The dual code $\C^\bot $ of  $\C$ has parameters $[(q-1)/2, (q-1)/2-2n,4]$. Moreover, the number of the minimum weight codewords in $\C^\bot $ is
 $$
 A_4(\C^\bot )= \frac{1}{8} \cdot (1 - 4\cdot 3^{2 m} + 3^{1 + 4 m}).
 $$
\end{enumerate}
\end{theorem}

\begin{table}[ht]
\begin{center}
\caption{The weight distribution of $\C$.} \label{tab-3}
\begin{tabular}{|c|c|}
\hline
weight & multiplicity \\[2mm]
\hline
$0$ & $1$
\\[2mm]
\hline
$3^{n-1}-3^m$ &  $\frac{1}{2} \cdot 3^m (1 + 3^m) (-1 + 3^{1 + 2 m})$
\\[2mm]
\hline
$3^{n-1}$ & $-1 + 2 \cdot 3^{1 + 4 m} + 9^m$
\\[2mm]
\hline
$3^{n-1}+3^m$& $\frac{1}{2} \cdot 3^m (-1 + 3^m) (-1 + 3^{1 + 2 m})$
\\[2mm]
\hline
\end{tabular}
\end{center}
\end{table}

\begin{proof}
(\uppercase\expandafter{\romannumeral1})~
By definitions, from the parameters of (\ref{c0}) we then deduce that
the code $\C$ has parameters $[(q-1)/2, 2n, 3^{n-1}-3^m]$ and three nonzero weights, i.e., $3^{n-1}-3^m$, $3^{n-1}$ and $3^{n-1}+3^m$. Further, from definitions and Lemma \ref{A230} we have that the minimum distance of $\C^\bot$ is at least $4$. Then the first three Pless power moments of
(\ref{eq:PPM}) yield the weight distribution in Table \ref{tab-3}.

(\uppercase\expandafter{\romannumeral2}) The desired conclusions follow from  the conclusion of (\uppercase\expandafter{\romannumeral1})  and the fifth Pless power moments of (\ref{eq:PPM}).
\end{proof}

\begin{rem}\label{rem0}
Note that the code $\C$ in Theorem \ref{thm-main1} and the ternary codes in \cite{Tangdcc2019} have the same parameters. However, these codes are different.

\end{rem}

\begin{example}\label{exam-31}
Let $m=2$. Then the code $\C$ is a $[121,10,72]$ ternary linear code with the weight enumerator
$1 + 10890 z^{72} +  39446 z^{81} +  8712  z^{90}.$  The code $\C$ is optimal. The dual code $\C^\perp $of $\C$ has parameters $[121,111,4]$ and is almost optimal according to the tables of best known codes maintained at http: //www.codetables.de. The number of the codewords of the minimum weight $4$ in $\C^\perp $ is $2420$.
\end{example}

\begin{example}\label{exam-32}
Let $m=3$. Then the code $\C$ is a $[1093 ,14 ,702]$ ternary linear code with the weight enumerator
$1 + 826308 z^{702} +  3189374 z^{729} +  767286  z^{756}$. The dual code $\C^\perp $of $\C$ has parameters $[1093,1079,4]$. The number of the codewords of the minimum weight $4$ in $\C^\perp $ is $198926$.
\end{example}

\section{Infinite families of $2$-designs from  ternary linear codes} \label{sec-des2}

In this section, we will show that the ternary code $\widetilde{\C_{n}}$ and its dual $\widetilde{\C_{n}}^\bot$  hold $2$-designs and  determine the parameters of these $2$-designs.
%To this end, we need the following lemma.

Let $m\geq 2$ be an positive integer and $d_0=3-2\cdot 3^{m+1}=-d^{-1} (mod~3^n-1)$. For convenience, we denote
$$
P=\{x\in \gf(q): ~x^{d_0}-1 \in SQ\}, ~NP=\{x\in \gf(q): ~x^{d_0}-1 \in NSQ\}
$$
$S=\{x\in \gf(q):~1-x^2\in SQ,~x\neq 0 \}$ and $NS=S^{-1}$. It is clear that $\gf(q)=S\cup NS \cup \gf(3)$. Then we have the following two results which will be used later. The former was documented in \cite{Dob2001}, and the latter is given in Lemma \ref{SZ}

\begin{lemma}\cite{Dob2001} \label{f}
Let symbols and notation be the same as before. Define $ f(x)=(x+1)^d-x^d$ with $x\in \gf(q)$. Then $f$ maps $S$, $NS$ and $\gf(3)$ into the set $P$, $NP$ and $\{1\}$, respectively. Moreover, the restriction of $f$ to $S$ maps $4$-to-$1$, and the
restriction of $f$ to $NS$ maps $2$-to-$1$.
\end{lemma}

%In order to show that the ternary code $\widetilde{\C_{n}}$ and its dual $\widetilde{\C_{n}}^\bot$  hold $2$-designs and  determine the parameters of these $2$-designs, the following result plays an important role in proving Theorem \ref{thm-2design}.

%To prove Theorem \ref{thm-2design}, we need the following result which plays an important role in determining the parameters of $2$-designs.

\begin{lemma}\label{SZ}
Let symbols and notation be the same as before. Let $z\in \gf(q) \setminus \gf(3)$ and denote
$$
S_{z,i}=\{(x,y)\in \gf(q)^2:~x+y+i z+1=0,~x^d+y^d+i z^d +1=0,~(x+1)(x+iz)\neq 0  \},
$$
where $i =\pm 1 $.
Then $(\# S_{z,+1}, \# S_{z,-1})=(2,0)~or~(0,2)$.
\end{lemma}

\begin{proof}
When $i=1$, we consider the following system
\begin{eqnarray}\label{N1}
\left\{
\begin{array}{ll}
x+y+z+1=0& \\[2mm]
x^d+y^d+z^d +1=0 &
\end{array}
 \right..
\end{eqnarray}
Substituting $x=-y-z-1$ into the second equation of (\ref{N1}) leads to
$$
(y+z+1)^d-y^d=z^d+1.
$$
Since $z \notin \gf(3)$, from the above equation we have
\begin{eqnarray}\label{N2}
(\frac{y}{z+1}+1)^d-(\frac{y}{z+1})^d=\frac{z^d+1}{(z+1)^d}
\end{eqnarray}
which can be written as
\begin{eqnarray}\label{N3}
f(\frac{y}{z+1})=\frac{z^d+1}{(z+1)^d},
\end{eqnarray}
where $f$ was defined in Lemma \ref{f}. Note that $y=-z$ and $y=-1$ are the two trivial solutions of (\ref{N2}). This means that (\ref{N3}) has solutions. Thus, from Lemma \ref{f} and $\frac{z^d+1}{(z+1)^d}\neq 1$ we have that the number of solutions of (\ref{N3}) is
\begin{eqnarray}\label{N24}
\left\{
\begin{array}{ll}
4 &  ~~~~\mbox{if} ~\left ( \frac{z^d+1}{(z+1)^d} \right )^{d_0}-1 \in SQ ~,\\[2mm]
2 &  ~~~~\mbox{if} ~\left ( \frac{z^d+1}{(z+1)^d} \right )^{d_0}-1 \in NSQ ~.
\end{array}
 \right.
\end{eqnarray}
Recall that $d=2\cdot 3^m+1$ and $d_0=3-2\cdot 3^{m+1}=-d^{-1}$. Then $dd_0=-1$ and

\begin{align}\label{Nd0}
\left ( \frac{z^d+1}{(z+1)^d} \right )^{d_0}-1
&=\frac{(z^d+1)^{d_0}}{(z+1)^{-1}}-1     \nonumber \\
&=\frac{z\cdot (z^{3^{m+1}+1}-1)^2}{(z^{3^{m+1}+2}+1)^2}.
\end{align}
Thus, the number of solutions of (\ref{N3}) is
\begin{eqnarray}\label{N24-1}
\left\{
\begin{array}{ll}
4 &  ~~~~\mbox{if} ~z\in SQ ~,\\[2mm]
2 &  ~~~~\mbox{if} ~z \in NSQ ~.
\end{array}
 \right.
\end{eqnarray}
This means that the number of solutions $(x,y)\in \gf(q)^2$ of (\ref{N1}) is equal to  (\ref{N24-1}). By the definition of $S_{z,+1}$ and removing the two trivial solutions $(x,y)=(-z,-1)$ and $(x,y)=(-1,-z)$, from (\ref{N24-1}) we have
\begin{eqnarray}\label{N24-2}
\# S_{z,+1}=
\left\{
\begin{array}{ll}
2 &  ~~~~\mbox{if} ~z\in SQ ~,\\[2mm]
0 &  ~~~~\mbox{if} ~z \in NSQ ~.
\end{array}
 \right.
\end{eqnarray}
Further, from definitions and  (\ref{N24-2}) we get
\begin{eqnarray*}
\# S_{z,-1}= \# S_{-z,+1}=
\left\{
\begin{array}{ll}
0 &  ~~~~\mbox{if} ~z\in SQ ~,\\[2mm]
2 &  ~~~~\mbox{if} ~z \in NSQ ~.
\end{array}
 \right.
\end{eqnarray*}
This completes the proof.
\end{proof}

Next we will give some infinite families of $2$-designs from ternary cyclic codes and determine the parameters of these designs.
\begin{theorem}\label{thm-2design}
Let symbols and notation be the same as before. Let $\C=\widetilde{\C_{n}}$ be defined by (\ref{c1}). Then we have the following results.
\begin{enumerate}
 \item [(\uppercase\expandafter{\romannumeral1})] The supports of the minimum weight codewords in $\C^\perp$ form a Steiner system $S(2,4,(q-1)/2)$.
 \item [(\uppercase\expandafter{\romannumeral2})] The code $\C$ and its dual $\C^\perp$ support $2$-designs. Furthermore, the minimum weight codewords of $\C$ support a simple $2$-$((q-1)/2,3^{n-1}-3^m,\lambda)$ designs with
\begin{eqnarray}\label{eq:numtamin1}
\lambda= \frac{3^{n-2}(3^{n-1}-3^m-1)}{2}.
\end{eqnarray}
\end{enumerate}

\end{theorem}

\begin{proof}
(\uppercase\expandafter{\romannumeral1})
Note that the dual code $\C^\perp$ has minimum distance $4$ from Theorem \ref{thm-main1}.
Let $x_3$ and $x_4$ be two distinct elements in $SQ$. Next we consider the number $N$ of $\{x_1,x_2\}\subseteq (SQ \setminus \{x_3,x_4\}) $ satisfying
\begin{eqnarray}\label{numA4}
\left\{
\begin{array}{ll}
c_1x_1+c_2x_2+c_3 x_3+c_4 x_4=0&\\[2mm]
c_1 x_1^{~d}+c_2 x_2^{~d}+c_3 x_3^{~d}+c_4 x_4^{~d}=0
\end{array}
 \right.,
\end{eqnarray}
where $c_1, c_2,c_3 $ and $c_4$ in $\gf(3)^*$.

Since the minimum distance of the dual code $\C^\perp$ is $4$,  there exist four elements $c_1, c_2,c_3 $ and $c_4$ in $\gf(3)^*$ and
four distinct elements $x_1, x_2,x_3 $ and $x_4$ in $SQ$ such that (\ref{numA4}) holds. Thus,
\begin{eqnarray}\label{A40}
N> 0.
\end{eqnarray}

Due to symmetry, we only consider the following three cases in (\ref{numA4}).
\begin{itemize}
  \item When $c_1=c_2=c_3=c_4=1$. Substituting $x=x_1/x_4$ , $y=x_2/x_4$ and $z=x_3/x_4$  into (\ref{numA4}) leads
to
\begin{eqnarray}\label{A41}
\left\{
\begin{array}{ll}
x+y+z+1=0& \\[2mm]
x^d+y^d+z^d +1=0 &
\end{array}
 \right..
\end{eqnarray}
  \item When $c_1=c_2=c_3=1$ and $c_4=-1$. Substituting $x=-x_1/x_4$ , $y=-x_2/x_4$ and $z=-x_3/x_4$  into (\ref{numA4}) leads
to (\ref{A41}).

  \item When $c_1=c_2=1$ and $c_3=c_4=-1$. Substituting $x=-x_1/x_4$ , $y=-x_2/x_4$ and $z=-x_3/x_4$  into (\ref{numA4}) leads
to
\begin{eqnarray}\label{A42}
\left\{
\begin{array}{ll}
x+y-z+1=0& \\[2mm]
x^d+y^d-z^d +1=0 &
\end{array}
 \right..
\end{eqnarray}
\end{itemize}
From Lemma \ref{SZ},  there exist two $(x,y)\in \gf(q)^2$ satisfying (\ref{A41}) or (\ref{A42}). Due to symmetry, it has only one set $\{x,y\} \subseteq  \gf(q)$ satisfying (\ref{A41}) or (\ref{A42}). This means that
\begin{eqnarray}\label{A43}
N\leq 1.
\end{eqnarray}
Combining (\ref{A40}) and (\ref{A43}), we get $N=1$. Therefore, the codewords of weight $4$ in $\C^\perp$ support a $2-((q-1)/2,4, 1)$ design, i.e., the supports of the minimum weight codewords in $\C^\perp$ form a Steiner system $S(2,4,(q-1)/2)$.

(\uppercase\expandafter{\romannumeral2}) By Theorem \ref{thm-main1}, we have that the code $\C$ has three nonzero weights, i.e.,
$w_1= 3^{n-1}-3^m$,  $w_2=3^{n-1} $ and $ w_3= 3^{n-1}+3^m $. It is clear that
$\left ( \mathcal P(\C), \mathcal B_{i}(\C) \right )$ are trivial $2$-designs for $i\in \{w_1, w_1+1,...,\frac{q-1}{2}-2\} \setminus \{w_1,w_2,w_3\}$. By Theorem \ref{thm-main1} and the conclusion (\uppercase\expandafter{\romannumeral1}), the minimum  distance of $\C^\perp$ is $4$ and the minimum weight codewords in $\C^\perp$ support $2$-designs. Thus, $\left ( \mathcal P(\C^\perp), \mathcal B_{i}(\C^\perp) \right )$ are $2$-designs for $0\leq i \leq 3+2-1$.
From Theorem \ref{thm-designGAMtheorem} we then deduce that both $\C$  and $\C^\perp$ hold $2$-designs. Moreover, for the minimum weight $i =3^{n-1}-3^m$ in $\C$ , the incidence structure $\left ( \mathcal P(\C), \mathcal B_{i}(\C) \right )$ is a simple $2$-$((q-1)/2, 3^{n-1}-3^m, \lambda)$ design with $b$ blocks, where
\begin{eqnarray}\label{eq:b}
b=
\frac{A}{3-1}=\frac{A}{2}
\end{eqnarray}
and $A$ is the number of the the minimum weight codewords in $\C$. Then the value of $\lambda$ in (\ref{eq:numtamin1}) follows from  Theorem \ref{thm-main1}, equations
(\ref{eq:bb}) and (\ref{eq:b}).
\end{proof}

\begin{rem}\label{rem}
Both the code $\C$ in Theorem \ref{thm-2design} and the ternary codes in \cite{Tangdcc2019}  support $2$-designs with the same parameters. However, these $2$-designs are not equivalent which confirmed by Magma programs.

\end{rem}

Let $\C=\widetilde{\C_{n}}$ be defined by (\ref{c1}). We regard $SQ$ as the set of the coordinate
positions $\mathcal P(\C)$ of $\C$.
By Theorems \ref{thm-main1} and  \ref{thm-2design}, from Theorem \ref{thm-PS} we can easily obtain the following result about shortened codes and omit its proof.

%We give some shortened codes of $\C$ and determine their parameters in the following theorem. The following result can be easily obtained from the parameters of $\C$  and we omit its proof.

\begin{theorem}\label{short1}
Let symbols and notation be the same as before. Let $m\geq 2$ be an integer and $T \subseteq SQ$. Then we have the following results.
\begin{itemize}
  \item If $\# T=1$, the shortened code $\C_{T}$ is a $[(q-3)/2, 2n-1,3^{n-1}-3^m]$ ternary linear code with the weight distribution in Table \ref{tab-short1}.
  \item If $\# T=2$, the shortened code $\C_{T}$ is a $[(q-5)/2, 2n-2,  3^{n-1}-3^m]$ ternary linear code with the weight distribution in Table \ref{tab-short2}.
\end{itemize}

\end{theorem}

\begin{table}[ht]
\begin{center}
\caption{The weight distribution of $\C_{T}$ for $\# T=1$.} \label{tab-short1}
\begin{tabular}{|c|c|}
\hline
weight & multiplicity \\[2mm]
\hline
$0$ & $1$
\\[2mm]
\hline
$3^{n-1}-3^m$ &  $\frac{1}{2} \cdot 3^m (-1 + 3^m + 3^{3 m} + 3^{1 + 2 m})$
\\[2mm]
\hline
$3^{n-1}$ & $-1 - 3^{2 m} + 2\cdot  3^{4 m}$
\\[2mm]
\hline
$3^{n-1}+3^m$&   $\frac{1}{2} \cdot 3^m (1 + 3^m + 3^{3 m} - 3^{1 + 2 m})$
\\[2mm]
\hline
\end{tabular}
\end{center}
\end{table}

\begin{table}[ht]
\begin{center}
\caption{The weight distribution of $\C_{T}$ for $\# T=2$.} \label{tab-short2}
\begin{tabular}{|c|c|}
\hline
weight & multiplicity \\[2mm]
\hline
$0$ & $1$
\\[2mm]
\hline
$3^{n-1}-3^m$ & $\frac{1}{2} \cdot 3^{ m-1} (-3 + 5 \cdot 3^m + 5 \cdot 3^{2 m} + 3^{3 m})$
\\[2mm]
\hline
$3^{n-1}$ & $\frac{1}{3} \cdot (-3 - 5 \cdot 3^{2 m} + 2\cdot  3^{4 m})$
\\[2mm]
\hline
$3^{n-1}+3^m$&  $\frac{1}{2} \cdot 3^{ m-1} (3 + 5 \cdot 3^m - 5 \cdot 3^{2 m} + 3^{3 m})$
\\[2mm]
\hline
\end{tabular}
\end{center}
\end{table}

\begin{example}\label{exa-short1}
Let $m=2$, $w$ be a primitive element of $\gf(3^5)$  and $T=\{w^2\}$. Then the shortened code $\C_{T}$  in Theorem \ref{short1} is a $[120,9,72]$ ternary linear code with the weight enumerator $1+4410 z^{72}+ 13040 z^{81}+ 2232 z^{90} $. The code $\C_{T}$ is optimal. The dual code of $\C_{T}$ has parameters $[120,111,3]$ and is almost optimal according to the tables of best known codes maintained at http: //www.codetables.de.
\end{example}

\begin{example}\label{exa-short2}
Let $m=3$, $w$ be a primitive element of $\gf(3^7)$  and $T=\{w^2, w^4\}$. Then the shortened code $\C_{T}$  in Theorem \ref{short1} is a $[1091, 12 ,702]$ ternary linear code with the weight enumerator $1+105570 z^{702}+ 353078 z^{729}+ 72792 z^{756} $. The dual code of $\C_{T}$ has parameters $[1091,1079,2]$.
\end{example}

\section{Pseudo-geometric design and the point-line design with parameters $S(2,4,121)$}\label{sec:pgc}

In this section, we will show that the Steiner system $S(2,4,121)$ presented in Section \ref{sec-des2} is inequivalent to the point-line design of the projective space $\mathrm{PG}(n-1,3)$ and thus is a pseudo-geometric design.

Throughout this section, let $\mathbb D_{n}$ denote the Steiner system derived from the
supports of the minimum weight codewords in $\widetilde{\C_{n}}^\perp$. Let $\mathrm{PG}_t(n-1,q)$  be the design whose points are
the points of $\mathrm{PG}(n-1,q)$, and whose blocks are the $t$-dimensional projective subspaces of $\mathrm{PG}(n-1,q)$.

We first mention
that the two incidence structures $\mathbb D_{n}$ and $\mathrm{PG}_1(n-1,3)$ are both cyclic
Steiner systems $S(2,4,\frac{3^n-1}{2})$. Since the smallest nontrivial Steiner system $S(2, 4, v)$ is
the unique $S(2, 4, 13)$ which is the design $\mathrm{PG}_1(2,3)$ of points and lines of $\mathrm{PG}(2, 3)$, our Steiner system $\mathbb D_{n}$ are equivalent to $\mathrm{PG}_1(n-1,3)$ when $n=3$. So
we consider the inequivalence only for the case $n\ge 5$.

The $p$-rank and block code can be used to distinguish inequivalent $t$-$(v,k,\lambda)$ designs.
Let $\mathbb  D$
be a
$t$-$(v,k,\lambda)$
 design with $b$ blocks.
 After numbering $v$ points and $b$ blocks in some way, respectively, we
define the incidence matrix of $\mathbb D$ design to be the matrix
\[M=(m_{i,j})_{1\le i\le b, 1\le j \le v}\]
where $m_{i,j}=1$ or $0$ according as the $j$th point occurs in the $i$th block or not.
The $p$-rank of $\mathbb D$ is defined as the rank
of $M$ over a field $\mathbb F$ of characteristic $p$, and it will be denoted by $\mathrm{rank}_{p}(\mathbb D)$. The
$\mathbb F$-vector space spanned by the rows of $M$ is called the (block) code of $\mathbb D$ over $\mathbb F$,
which is denoted by $\mathcal C_{\mathbb F} (\mathbb D)$. If $\mathbb F=\gf(q)$, where $q$ is a power of $p$, then we denote the code of $\mathbb D$ over $\mathbb F_q$ by $\mathcal C_{q} (\mathbb D)$. The $p$-ranks of the finite geometry designs have been studied since the 1950's, due to interest
in the dimensions of these majority-logic decodable codes.
In 1992, Ceccherini and Hirschfeld \cite{CH92} found a formula for the $p$-rank of all point-line designs from projective spaces of prime order:
\[\mathrm{rank}_{p} \left ( \mathrm{PG}_1(n-1,p) \right ) =\frac{p^n-1}{p-1}- \binom{n+p-2}{p-1}.\]
For the special case $p=3$ and $n=5$, we have $\mathrm{rank}_{3} \left ( \mathrm{PG}_1(4,3) \right )=106$. With the help
of Magma, we obtain $\mathrm{rank}_{3} (\mathbb D_5)=111$.  This shows that the Steiner system $\mathbb D_n$ is inequivalent to the classical point-line design  $\mathrm{PG}_1(n-1,3)$ when $n=5$ and thus $\mathbb D_5$ is a pseudo-geometry design with the parameter $S(2,4,121)$. Moreover, the two codes $\mathcal C_3(\mathbb D_5)$ and $\widetilde{\mathcal C_5}^{\perp}$ are identical.

\begin{open}
Let $n$ be an odd integer. Prove or disprove the inequivalence between
	 $\mathbb D_n=(\mathcal P( \widetilde{\mathcal C_n}^{\perp}),\mathcal B_4(\widetilde{\mathcal C_n}^{\perp}))$
	 and the classical point-line design  $\mathrm{PG}_1(n-1,3)$ in the case $n\ge 7$.
\end{open}

Let $\alpha$ be the primitive element of $\mathrm{GF}(3^n)$ in (\ref{c1}). Consider the ternary cyclic code
\begin{eqnarray*}
{\C_{n}}=\left \{(\tr(a \alpha^{2i}+b\alpha^{4i}))_{i=0}^{(3^n-3)/2}|a,b\in \gf(3^n)\right \}.
\end{eqnarray*}
Then the incidence structure $(\mathcal P(\mathcal{C}_n^{\perp}),\mathcal B_4({\mathcal C_n}^{\perp}))$ is isomorphic
to the point-line design $\mathrm{PG}_1(n-1,3)$ (see \cite{Tangdcc2019}). With the help
of Magma, we have $\mathcal B_4({\mathcal C_n}^{\perp}) \bigcap \mathcal B_4(\widetilde{\mathcal C_n}^{\perp})=\emptyset$ when $n=5$. Thus, the two Steiner systems $(\mathcal P( \widetilde{\mathcal C_5}^{\perp}),\mathcal B_4(\widetilde{\mathcal C_n}^{\perp}))$ and $(\mathcal P(\mathcal{C}_n^{\perp}),\mathcal B_4({\mathcal C_n}^{\perp}))$   are disjoint when $n=5$.

\begin{open}
Let $n$ be an odd integer. Prove or disprove the disjointness between $(\mathcal P( \widetilde{\mathcal C_n}^{\perp}),\\~\mathcal B_4(\widetilde{\mathcal C_n}^{\perp}))$ and $(\mathcal P(\mathcal{C}_n^{\perp}),\mathcal B_4({\mathcal C_n}^{\perp}))$ in the case $n\ge 7$.
\end{open}

An automorphism group of a $t$-design gives information about the structure of the
design, but it is difficult to be determined. Computer calculations show that the automorphism group of $\mathbb D_5=(\mathcal P( \widetilde{\mathcal C_5}^{\perp}),\mathcal B_4(\widetilde{\mathcal C_5}^{\perp}))$  is isomorphic to $C_{121} \rtimes \mathrm{Gal}(\mathrm{GF}(3^5)/\mathrm{GF}(3))$ and the automorphism group of
$\mathrm{PG}_1(4,3)$ is isomorphic to the projective general linear group $\mathrm{PGL}(5,3)$, i.e., the group corresponding to the action of $\mathrm{GL}(5,3)$ on the projective points of the $5$-dimensional vector space $\mathrm{GF}(3)^5$.
\begin{conj}
Let $n\ge 7$ be an odd integer. Then the automorphism group of $\mathbb D_n=(\mathcal P( \widetilde{~\mathcal C_n}^{\perp}),~\\ \mathcal B_4(\widetilde{\mathcal C_n}^{\perp}))$  is isomorphic to $C_{(3^n-1)/2} \rtimes \mathrm{Gal}(\mathrm{GF}(3^n)/\mathrm{GF}(3))$, and the automorphism group of
$\mathrm{PG}_1(n-1,3)$ is isomorphic to the projective general linear group $\mathrm{PGL}(n,3)$.
\end{conj}

\section{Summary and concluding remarks}\label{sec-summary}
In this paper, we mainly investigated a class of ternary cyclic codes related to the ternary $m$-Sequences with
Welch-type. The parameters of those ternary codes and their dual codes were also determined. Meanwhile, we showed that those ternary codes held $2$-designs. Especially,
the minimum distance of the dual of the studied ternary codes is $4$ and the supports of those codewords with weight $4$ form a Steiner system $S(2,4,\frac{3^n-1}{2})$.
We showed that this Steiner system $S(2,4,\frac{3^n-1}{2})$ is not isomorphic to $\mathrm{PG}_1(n-1,3)$ and yielded a pseudo-geometric design when $n=5$.  Furthermore, the parameters of some shortened codes of the studied ternary codes were also given.

%We showed that the Steiner system $S(2,4,\frac{3^n-1}{2})$ presented in this paper is not isomorphic to $\mathrm{PG}_1(n-1,3)$ when $n=5$, although we were unable to prove that these Steiner systems $S(2,4,\frac{3^n-1}{2})$ are inequivalent to $\mathrm{PG}_1(n-1,3)$ when $n \ge 7$. Furthermore, the parameters of some shortened codes of the studied ternary codes were also given.

We remark that the parameters of the obtained $2$-designs in this paper are the same as that of the $2$-designs in \cite{Tangdcc2019}, but those $2$-designs are not equivalent.
This paper does not consider the cases (\uppercase\expandafter{\romannumeral1}) and (\uppercase\expandafter{\romannumeral2}) which listed in Section \ref{sec-int}, since Magma program shows that the corresponding ternary codes and the 2-designs supported by these codes  are equivalent to that of \cite{Tangdcc2019} when $n=5$. To conclude this paper, we further presents the following
conjecture.
\begin{conj}
Let $m\geq 2$ be an positive integer, $p=3$ and $n=2m+1$. Let $\C=\widetilde{\C_{n}}$ be defined by (\ref{c1}) with $(n,d)$ satisfying (\uppercase\expandafter{\romannumeral1}) or (\uppercase\expandafter{\romannumeral2}) which listed in Section \ref{sec-int}. Then the ternary code $\C$ and its dual code $\C^\perp$ hold $2$-designs. Moreover, these ternary codes and 2-designs are equivalent to that of \cite{Tangdcc2019}.
\end{conj}

%since Magma program shows that the corresponding incidence structures are not 2-designs. To conclude this paper, we further presents the following conjectures:
%some of those codes are optimal in the sense that their parameters meet certain bound of linear codes.

\section*{Acknowledgements}

This research was supported by the National Natural Science Foundation of China under grant numbers 12171162, 11971175 and 12231015, the Sichuan Provincial Youth Science and Technology Fund under grant number 2022JDJQ0041, the Innovation Team Funds of China West Normal University under grant number KCXTD2022-5, the Natural Science Foundation of Sichuan Province under grant number 2022NSFSC1805 and the Fundamental Research Funds for the Central Universities of China under grant number 2682023ZTPY002.

%The authors are very grateful to the reviewers and the Editor, for their comments and suggestions that improved the presentation and quality of this paper. C. Xiang's research was supported by the National Natural Science Foundation of China under grant numbers 12171162 and 11971175, and the Basic Research Project of Science and Technology Plan of Guangzhou city of China under grant number 202102020888. C. Tang's research was supported by the National Natural Science Foundation of China under grant number 12231015, the Sichuan Provincial Youth Science and Technology Fund under grant number 2022JDJQ0041 and  the Innovation Team Funds of China West Normal University under grant number KCXTD2022-5.

\end{document}